\documentclass[11pt]{article}

\usepackage{natbib}

\usepackage{amsmath}
\usepackage{url}

\usepackage{star}

\usepackage{times}
\usepackage{bm}

\usepackage{algorithm} 
\usepackage{algpseudocode} 

\def\T{{ \mathrm{\scriptscriptstyle T} }}

\usepackage{booktabs}
\usepackage[flushleft]{threeparttable}

\usepackage{graphics}

\usepackage{xr}
\makeatletter
\newcommand*{\addFileDependency}[1]{
  \typeout{(#1)}
  \@addtofilelist{#1}
  \IfFileExists{#1}{}{\typeout{No file #1.}}
}
\makeatother
 
\newcommand*{\myexternaldocument}[1]{%
    \externaldocument{#1}%
    \addFileDependency{#1.tex}%
    \addFileDependency{#1.aux}%
}

\myexternaldocument{supp}

\usepackage[justification=centering]{caption}
\usepackage{tikz-cd}

\graphicspath{{Figures_manuscript/}}

\usepackage{array}
\newcolumntype{L}[1]{>{\raggedright\let\newline\\\arraybackslash\hspace{0pt}}m{#1}}
\newcolumntype{C}[1]{>{\centering\let\newline\\\arraybackslash\hspace{0pt}}m{#1}}
\newcolumntype{R}[1]{>{\raggedleft\let\newline\\\arraybackslash\hspace{0pt}}m{#1}}

\usepackage{appendix}

\usepackage{bbm}

\usepackage{tabularx}

\newcommand\clearrow{\global\let\rowmac\relax}
\clearrow

\DeclareRobustCommand\sampleline[1]{%
  \tikz\draw[#1] (0,0) (0,\the\dimexpr\fontdimen22\textfont2\relax)
  -- (2em,\the\dimexpr\fontdimen22\textfont2\relax);%
}

\usepackage{amssymb}
\newcommand\Z{\mathbb{Z}}
\newcommand\D{\mathbb{D}}

\newcommand{\Y}{\mathbb{Y}}
\newcommand{\cov}{\mathrm{cov}}

\newcommand*{\var}{\mathrm{var}}
\newcommand*{\pr}{\mathrm{pr}}
\newcommand\U{\mathbb{U}}

\usepackage{subfigure}

\usepackage{mathtools}

\DeclarePairedDelimiter\floor{\lfloor}{\rfloor}

\usepackage{url}

\usepackage{appendix}

\usepackage{multirow}

\makeatletter
\newcommand{\customlabel}[2]{%
\protected@write \@auxout {}{\string \newlabel {#1}{{#2}{}}}}
\makeatother

\providecommand{\keywords}[1]
{
  \small	
  \textit{Keywords:} #1
}

\begin{document}

\date{}

\title{Fighting Noise with Noise: Causal Inference with Many Candidate Instruments}

\author{Xinyi Zhang$^{1}$ , Linbo Wang$^2$ , Stanislav Volgushev$^2$ and Dehan Kong$^2$ \vspace{1mm} \\
$^1$ Department of Biostatistics, Johns Hopkins University,  Baltimore, USA \vspace{1mm} \\
 $^2$Department of Statistical Sciences, University of Toronto,  Toronto,  Canada}

\maketitle

\begin{abstract}
Instrumental variable methods provide useful tools for inferring causal effects in the presence of unmeasured confounding. To apply these methods with large-scale data sets, a major challenge is to find valid instruments from a possibly large candidate set. In practice, most of the candidate instruments are often not relevant for studying a particular exposure of interest.  Moreover, not all relevant candidate instruments are valid as they may directly influence the outcome of interest. In this article, we propose a  data-driven method for causal inference with many candidate instruments that addresses these two challenges simultaneously. A key component of our proposal involves using pseudo variables, known to be irrelevant, to remove variables from the original set that exhibit spurious correlations with the exposure. Synthetic data analyses show that the proposed method performs favourably compared to existing methods. We apply our method to a Mendelian randomization study estimating the effect of obesity on health-related quality of life.
\end{abstract}

\hspace{2mm}

\section{Introduction}\label{introduction}
The instrumental variable model is a workhorse in causal inference using observational data. It is especially useful when standard adjustment methods are biased due to unmeasured confounding. The key idea is to find an exogenous variable to extract random variation in the exposure and use this random variation to estimate the causal effect. This exogenous variable is known as an instrumental variable or instrument if it is related to the exposure, but does not affect the outcome  except through the exposure. 

In the past, the selection of appropriate exogenous variables that can serve as valid instruments was typically driven by expert knowledge \citep[e.g.][]{angrist1991does}. In recent years, applied researchers have increasingly sought to employ large-scale observational data sets to identify causal relationships. As the size and complexity of data sets grow, it becomes increasingly difficult to build a knowledge-based instrumental variable model for modern data applications.  A prominent example in genetic epidemiology is Mendelian randomization, which uses genetic variants, such as single-nucleotide polymorphisms (SNPs), as instruments to assess the causal relationship between a risk factor and an outcome  \citep[e.g.][]{Burgess2015}.  The candidate instrument set includes millions of SNPs in the human genome. This challenge  also arises in the technology domain, particularly when assessing the impact of recommendation systems of e-commerce platforms like online bookstores. In this context, the occurrence of a large and sudden increase in demand for a specific product may serve as a binary instrument  \citep[e.g.][]{Carmi2012}. With extensive record data and numerous products, a large collection of such events may qualify as candidate instruments.

Modern large-scale complex systems call for data-driven methods that find valid instruments from a high-dimensional candidate set. Two main challenges arise from such a pursuit. First, often most of the candidate instruments are not relevant for studying a particular exposure of interest.  To deal with this problem, a common approach is to employ a screening procedure \citep[e.g.][]{JianqingFan2008} to select variables in the candidate set that are associated with the exposure. For example, in genetics, it is common to use genome-wide association studies (GWAS) or variants thereof
to identify genetic variants that are relevant to a risk factor.  Recent studies have suggested choosing a less conservative selection threshold than the traditional $5\times 10^{-8}$ for Mendelian randomization applications \citep[e.g.][]{Zhao2020}, in order to 
include more SNPs that are associated with the risk factors and hence improve the efficiency of the resulting causal effect estimate; see also  \cite{EmmanuelJ.Cands2019} for a similar recommendation in a related context. This is especially important for datasets with small to moderate sample sizes, for which  standard GWAS has limited power so that only very few or even no genetic variants may be selected. 
However, a less conservative selection threshold will result in more false positives, i.e., SNPs having spurious correlations with the risk factor of interest. Using these selected instruments will undermine the validity of downstream causal examinations.

Moreover, some of the candidate instruments found to be associated with the exposure might not be valid as they may directly impact the outcome of interest. For example, in Mendelian randomization studies, this problem arises due to pleiotropy, a phenomenon that one SNP influences multiple traits and possibly through independent pathways. In response to this concern, recently there has been a surging interest in methods that make valid causal inference in the presence of invalid instruments. \cite{Kolesar2015} and \cite{Bowden2015}  provided inferential methods for treatment effects in the presence of invalid instruments. Their methods assume that the direct effects of these invalid instruments are uncorrelated with the associations between the instruments and the exposure. 
Without knowing a set of valid instruments a priori, the majority rule \citep[e.g.][]{Bowden2016,Kang2016}--at least 50\% of the candidate instruments are valid--is considered for consistent estimation of valid instruments. \cite{Guo2018} and \cite{Windmeijer2019b} further relax this condition and assume that the largest group of Wald ratios with equal values corresponds to valid instruments, referred to as the plurality rule.

The developments in the present paper are motivated by the surprising finding that a naive combination of existing methods for solving the two challenges discussed earlier leads to biased causal inference. More precisely, we observe that estimates from this naive combination are close to the ordinary least squares estimate that does not account for any unmeasured confounding.

Our contributions in the paper are two-fold.  We first provide a theoretical characterization of the observation described above. In particular, we show that estimates from exogenous variables having spurious correlations with the exposure are concentrated in a region that can be separated from the true causal effect. Remarkably in our theoretical analysis, we explicitly take into account the randomness in the first  marginal screening step, while most of the existing literature treats this as pre-processing and does not account for its randomness. 

We then introduce a novel strategy that addresses this challenge. A key component of our proposal is the use of pseudo variables, known to be irrelevant, to identify and remove spurious variables. The idea of pseudo variables has been investigated before in other contexts, most notably in forward selection \citep{Wu2007}. It is also similar in spirit to knockoffs \citep{RinaFoygelBarberan2015} although  our construction of pseudo variables is much simpler. 
To facilitate inference, we employ a sample splitting procedure with theoretical justifications. 
One portion of the data is used to detect spurious variables, and the remaining data are used to identify valid instruments and estimate the causal effect. 
Our method guarantees the removal of all spurious instruments and the correct identification of all valid instruments. Additionally, the causal estimate is demonstrated to be asymptotically normal.

\section{Preliminary}

\subsection{Notation}\label{notation_main}
Throughout, 
for a vector $v\in \mathbb{R}^p$ and a set $\mathcal{S} \subset \{1, \ldots, p\}$ with cardinality $|\mathcal{S}|$, $ v_{\mathcal{S}}$ denotes a subvector of $v$ indexed by $ \mathcal{S} $ and $ v_{\mathcal{S}, l} $ is the $l$th element of $v_{\mathcal{S}}$ for $1\leq l \leq |\mathcal{S}|$. Let $ \|  Z \|_{\psi_2}$ denote the sub-Gaussian norm of a random vector $Z \in \mathbb{R}^p$. For a matrix $A\in \mathbb{R}^{p \times p}$,  
$\lambda_{\min}(A)$ and $ \lambda_{\max}(A)$ are the smallest and the largest eigenvalues of $A$, respectively.  
Let $[A]_{-j,j}$ denote the $j$th column of $A$ without the $j$th component.
The matrix $A_{\mathcal{S}}$ is formed by rows and columns of $A$ indexed by $\mathcal{S}$ and its inverse is denoted by $ A_{\mathcal{S}}^{-1} $.

\subsection{Background}\label{models}
The instrumental variable approach is a popular method for inferring causality from observational studies. It introduces exogenous variables called instruments to control for unmeasured confounding $U$.
Conditioning on observed covariates, a variable $Z$ is called an instrument for estimating the causal effect of an exposure $D$ on an outcome $Y$ if it satisfies the three core conditions \citep[e.g.][]{VanessaDidelez2007, LinboWang2018}:
(I1) Relevance: $Z$ is associated with $D$; (I2) Exclusion restriction: $Z$ has no direct effect on $Y$;
(I3) Independence: $Z$ is independent of unmeasured variables that affect $D$ and $Y$.

Suppose we have $n$ independent samples $(Y_i, D_i, {Z}_{i},   X_{i})$ from the distribution of $ (Y, D,  {Z},   X) $, where $Y$ is a continuous outcome of interest, $D$ is a continuous exposure, $  Z=(Z_1, \ldots, Z_p)^{\T} \in  \mathbb{R}^{p}$ denotes candidate instrumental variables, and $  X=(X_1, \ldots, X_q)^{\T} \in \mathbb{R}^q$ represents baseline covariates. We assume that $D$ and $Y$ are generated from the following linear structural equation models: 
\begin{align}
D &=   {Z}^{\T}  {\gamma}^* +     X^{\T}    \psi^*+     U^{\T}    \alpha_D^* + \epsilon_{D}   \label{modD_XZ};\\
Y &=   {Z}^{\T}  {\pi}^* +  \beta^* D +    X^{\T}    \phi^* +    U^{\T}    \alpha_Y^*  + \epsilon_{Y} \label{modY_XZ}.
\end{align}
The parameter $\beta^*$ measures the causal effect of the exposure $D$ on the outcome $Y$, which is the quantity of primary interest. Parameter $  \gamma^* \in \mathbb{R}^p$ characterizes the partial correlations between the instruments and the exposure. Parameter $   \pi^* \in \mathbb{R}^{p}$ measures the degree of violation of the exclusion restriction assumption (I2). Similarly, parameters $    \psi^*,    \phi^* \in \mathbb{R}^{q}$ refer to the effects of observed covariates on the exposure and the outcome, respectively. 
Parameters  $    \alpha_D^* \in \mathbb{R}^g$ and $   \alpha_Y^* \in \mathbb{R}^g$ represent the direct effects of the multivariate unobserved confounders $   U \in \mathbb{R}^{g}$ on $D$ and $Y$, respectively. The dimension $g$ is fixed. Unmeasured confounding $   U$ and random errors $\epsilon_D, \epsilon_Y$ are independently distributed with mean zero and 
variance $\Sigma_U, \sigma_D^2, \sigma_Y^2$, respectively. We further assume that $Z_j$ is independent of $\epsilon_D, \epsilon_Y$ and $   U$ for $j = 1, \ldots, p$. 
An instrument $Z_j$ is called relevant if $\gamma_j^* \neq 0$, and is called irrelevant otherwise. For relevant instruments, if $\pi_j^*= 0$, then $Z_j$ is a valid instrument; otherwise, it is called an invalid instrument.

For the rest of the paper, we suppress dependence on covariates $   X$  for simplicity. We consider \begin{align}
&D=   {Z}^{\T}  {\gamma}^* +    U^{\T}    \alpha_D^* + \epsilon_D, \label{D_reduce} \\
&Y=  {Z}^{\T}  {\Gamma}^* +   U^{\T} (\beta^*    \alpha_D^*+  \alpha_Y^*) + e,  \label{Y_reduce}
\end{align}
where $  {\Gamma}^*=  {\pi}^*+\beta^* {\gamma}^*$ and $e = \beta^* \epsilon_{D}+\epsilon_{Y}$. Let $ \Z \in \mathbb{R}^{n\times p},   \D \in \mathbb{R}^n$ and $  \Y \in \mathbb{R}^n$ denote the collection of observations for $( Z, D, Y)$ from $n$ independent units, respectively. 
Here $  \Z$ is also called the design matrix. 
Without loss of generality, assume that  $  \Z,   \D,   \Y$ are all centered.

Given a valid instrument $Z_j,$ the causal effect $\beta^*$ may be identified as the Wald ratio \citep{Lawlor2008}.
Alternatively, given a set of valid instruments $ Z_{valid}$, a popular procedure to estimate the causal effect under linear models is the two-stage least squares  estimator given by
$
\widehat{\beta}_{\text{2SLS}} = (  {\D}^{\T} P   \D)^{-1}(   \D^{\T} P   \Y),
$
where 
$
P=    \Z_{valid} (    \Z_{valid}^{\T}     \Z_{valid})^{-1}   \Z_{valid}^{\T}
$
is a projection matrix and $   \Z_{valid}$ is the design matrix corresponding to the valid instruments.

In practice, prior knowledge of instrument validity is often not available. Nevertheless, the causal effect can still be identified and estimated from data under 
the plurality rule, which assumes that the largest group of invalid instruments with the same Wald ratio is smaller than the group of valid instruments, whose Wald ratio equals the true causal effect. In this case,  the mode of the Wald ratios constructed using individual candidate instruments corresponds to the true causal effect  \citep[e.g.][] {Guo2018, Windmeijer2019b}.

\section{Challenges for Causal Inference with Many Candidate Instruments}\label{challenge_highD}

\subsection{A naive approach}\label{naive_combo}

Causal effect estimation under models \eqref{modD_XZ} and \eqref{modY_XZ}, where the number of candidate instruments far exceeds the sample size, is an important practical problem. 
For example, in the genetics application detailed in Section \ref{real_data}, there are 3,683,868 candidate instruments measured on a dataset of 3023 subjects. A naive method to addressing this challenge is to first reduce dimension 
and then apply existing procedures that assume the plurality rule to the remaining instruments. One version of this approach, based on the proposal of \cite{Guo2018}, is briefly described below and outlined in Algorithm \ref{naive_algorithm} in the supplement, with additional details in Section \ref{naiveAlgo_detail}.

To alleviate computational demands in ultra-high dimensional settings, marginal screening is commonly used to quickly screen out variables unlikely to be relevant to the exposure \citep[e.g.][]{JianqingFan2008}.
In the present setting, we select the top $s$ candidate instruments by ranking the absolute marginal estimates $
  |\sum_{i=1}^{n} D_i Z_{ij}|/(\sum_{i=1}^n Z_{ij}^2)
$ for $1\leq j \leq p$. 
Let $\widehat{\mathcal{S}}_1 \subset \{ 1, \ldots, p \}$ be the index set of the top $s$ candidate instruments. 

We then apply the method of \cite{Guo2018} to the selected instruments $\widehat{\mathcal{S}}_1$.  To be more specific, given that $s$ is potentially large relative to $n$, the joint relevance strength of candidates $ Z_{\widehat{\mathcal{S}}_1}$ is evaluated using a regularization-based approach such as the de-biased lasso estimator of \cite{VANDEGEER2014}, as
detailed in Section \ref{M_construct} of the supplement. The joint estimates regarding $D$ and $Y$ are denoted by $\widehat{    \gamma}_{\widehat{\mathcal{S}}_1}$ and $\widehat{    \Gamma}_{\widehat{\mathcal{S}}_1}$, respectively, and $ \gamma_{\widehat{\mathcal{S}}_1}^*, \Gamma_{\widehat{\mathcal{S}}_1}^* =    {\pi}_{\widehat{\mathcal{S}}_1}^* + \beta^*    {\gamma}_{\widehat{\mathcal{S}}_1}^*$ represent their true values. Next perform a hard thresholding  \citep{Donoho1994},  and  let $ \widehat{\mathcal{S}}_2=\{1\leq l \leq s: |\widehat{\gamma}_{\widehat{\mathcal{S}}_1, l}| \geq \delta_n\times \text{SE}(\widehat{\gamma}_{\widehat{\mathcal{S}}_1,l}) \}$ denote the estimated relevant instruments. Here $\delta_n = \surd{\left[\omega \log\left\{ \max\left(n,s\right)\right\}\right]}$, with $\omega$ as a tuning parameter, and $\text{SE}(\widehat{\gamma}_{\widehat{\mathcal{S}}_1, l})$ is the standard error of $\widehat{\gamma}_{\widehat{\mathcal{S}}_1, l}$.  The explicit form of $\text{SE}(\widehat{\gamma}_{\widehat{\mathcal{S}}_1, l})$ is given in the Supplement Section \ref{SE_debiasLasso}. 
Applying the voting framework \citep{Guo2018} to $\widehat{\mathcal{S}}_2$ identifies valid instruments, denoted by $\widehat{\mathcal{S}}_4$.

\subsection{A numerical example}\label{challenges}

In this section, we analyze the naive approach from Section \ref{naive_combo} in a simulated example. We generate 1000 data sets of size $n=500$ with  $p=50,000$ candidate instruments. The true causal effect $\beta^*$ is set to be 2.  Candidate instruments $Z_1, Z_2$ are invalid, $Z_3, \ldots, Z_9$ are valid, and $Z_{10}, \ldots, Z_{50,000}$ are irrelevant. The valid instruments satisfy $ Z_{3:9}\sim N( 0, \Sigma_v)$ with $ [\Sigma_{v}]_{jk} = 0.25^{|j-k|}$. The other candidate instruments $Z_j, j = 1,2,10, \ldots, 50,000$,  observed covariates $X_1, X_2$,  random error $\epsilon_Y $, and a univariate unobserved confounding $U$ are generated from independent standard normal distributions. 
The treatment $D$ and the outcome $Y$ follow models \eqref{modD_XZ} and \eqref{modY_XZ} with $(\alpha_D^*, \alpha_Y^*)= (4,-3),  \psi^* = (1.5, 2)^{\T}, \phi^* = (1.2, 1.5)^{\T}$; $ \gamma_{1:9}^*=3, \pi_{1:2}^* = (-3.5,3.5)^{\T}$ and all other elements of $\gamma^*, \pi^*$ are zeros. 
For simplicity, we let $\epsilon_D=0$.

Following the naive algorithm presented in Section \ref{naive_combo}, on average across 1000 Monte Carlo runs, 26.21 variables are estimated as valid instruments, among which 26.20 are in fact irrelevant.
To provide insights into this problem, we study the causal effect estimates from $Z_j$s with $j \in \widehat{\mathcal{S}}_2$,  the set of candidates passing joint thresholding. 
In Figure \ref{fig-S1X-sigma1-motiv}(a), we plot the histograms of causal effect estimates aggregated over 1000 Monte Carlo runs. We use different colors to distinguish the estimates produced by the valid, invalid instruments and irrelevant variables.
One can see that the plurality rule  is violated empirically even after joint thresholding: the largest number of causal effect estimates with similar values corresponds to irrelevant variables. The latter variables shall be referred to as spurious instruments. The plurality rule is violated here because there are many more irrelevant candidate instruments to begin with, and the sample size is small relative to the number of candidate instruments. Somewhat surprisingly,  all spurious instruments lead to similar causal effect estimates. Moreover, these effect estimates lie in a region separated from those estimates corresponding to the valid instruments. As a result, the naive algorithm that assumes the plurality rule misidentifies the spurious instruments as valid instruments, as shown in Figure \ref{fig-S1X-sigma1-motiv}(b). A theoretical explanation of these findings is given in the following subsection.

\begin{figure}[htbp]
\centering
\subfigure[$\widehat{\mathcal{S}}_2$
]{\includegraphics[width=0.3\textwidth]{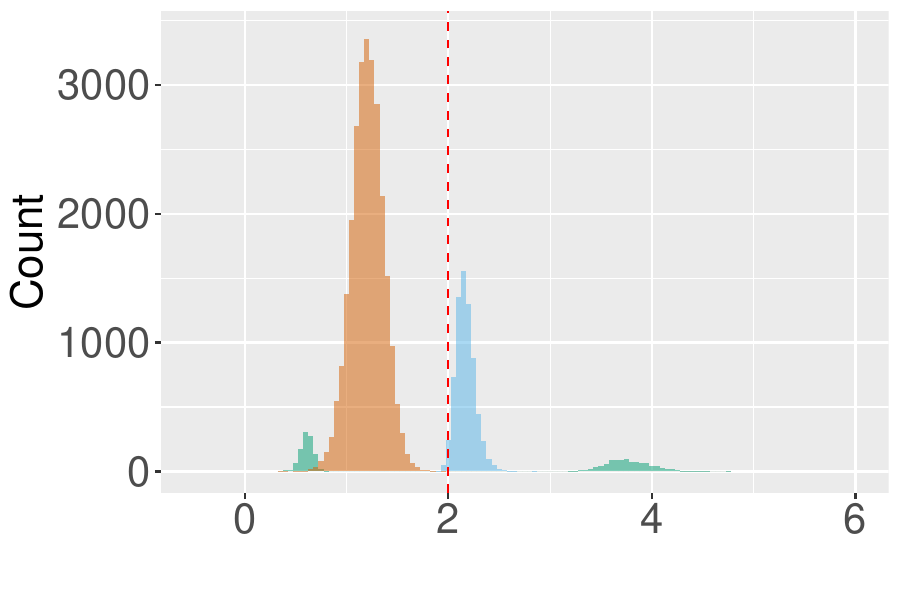} 
}
\hspace{0.5cm}
\subfigure[$\widehat{\mathcal{S}}_4$
]{\includegraphics[width=0.3\textwidth]{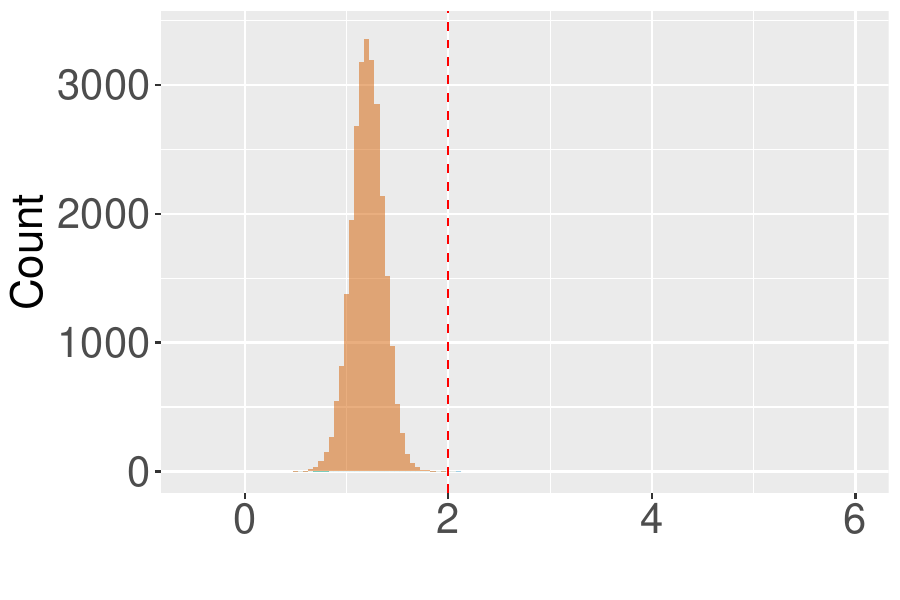}
}
\vspace{-0.4cm}
\caption{Histograms of causal effect estimates $\widehat{\Gamma}_{\widehat{\mathcal{S}}_1, l}/\widehat{\gamma}_{\widehat{\mathcal{S}}_1, l}$ for $l=1, \ldots, s$ from valid (blues), invalid (green) and spurious (orange) instruments aggregated over 1000 Monte Carlo runs. (a) plots the estimates by candidate instruments that pass the joint thresholding step ($\widehat{\mathcal{S}}_2$) of the naive algorithm, 
and (b) plots the estimates by candidate  instruments identified as valid instruments ($\widehat{\mathcal{S}}_4$) by the naive algorithm.  
The red dashed line is the true causal effect $\beta^*=2$.}
\label{fig-S1X-sigma1-motiv}
\end{figure}

\subsection{Concentration results for causal estimates from variables passing joint thresholding}

In this section, we theoretically characterize the behavior of causal effect estimates from valid, invalid instruments that pass joint thresholding, and  irrelevant variables that pass joint thresholding. The latter are referred to as spurious instruments.
We define the sets of relevant and irrelevant variables selected by marginal screening as
    $\check{\mathcal{R}} = \left\{l: 1\leq l \leq s, \gamma_{\widehat{\mathcal{S}}_1,l}^* \neq 0 \right\}$, 
    $
    \check{\mathcal{I}} = \left\{l: 1\leq l \leq s, \gamma_{\widehat{\mathcal{S}}_1,l}^* = 0 \right\}. $
We impose the following assumptions:
\vspace{-2mm}
\begin{itemize}
\item[]A1. Unmeasured confounders $  U_i$'s are $i.i.d.\text{ } N(  0, \Sigma_U)$ where $0 \in \mathbb{R}^g$, $\Sigma_U \in \mathbb{R}^{g\times g}$ and $g$ is fixed. Random errors $\epsilon_{iD}$'s are $i.i.d. \text{ } N(0, \sigma_D^2)$ and $\epsilon_{iY}$'s are $i.i.d.\text{ } N(0, \sigma_Y^2)$. Assume $\epsilon_{iD}, \epsilon_{iY},  U_i, Z_{ij}$ are independent for $1\leq i \leq n$ and $1\leq j \leq p$. \customlabel{assumption1}{A1}
\item[]A2. Let $   \Sigma  = \cov( Z) \in \mathbb{R}^{p \times p} $ . There exist constants $C_L, C_R, B$ such that for all $p$, 
$0< C_L \leq \lambda_{\min}(  \Sigma) \leq \lambda_{\max}( \Sigma) \leq C_R < \infty$ and $\max_{1\leq j \leq p}   \Sigma_{jj}\leq B$. 
Assume $  Z_i $s are $i.i.d$ sub-gaussian with variance proxy $\sigma^2$ and $E(Z_i)=0$. 
Let $\kappa = \|   \Sigma^{-1/2}   Z_1\|_{\psi_2} < \infty$. \customlabel{assumption2}{A2}
\item[]A3. Irrelevant $Z_{k}$s are irrelevant to both exposure and outcome: $\{k: \gamma_k^* = 0\}\subseteq \{k: \pi_k^* = 0\}$. 
\item[]A4. The number of relevant $Z_j$s to the exposure $D$, denoted by $s_1^* $, is fixed. The number of  $Z_j$s that violate the exclusion restriction assumption (I2) is denoted by $s_2^*$. \customlabel{assumption4}{A4}
\item[]A5. There exist constants $c_0, C_0, \breve{c}_0, \breve{C}_0>0$ such that $c_0\leq  |\gamma_{k}^*| \leq C_0$ for all relevant variables $k$ and for all invalid instrumenst $k $, the degree of violation satisfies $\breve{c}_0\leq |\pi_{ k}^*|\leq \breve{C}_0$. There exists $c>0$ such that  $\min_{k: \gamma_k^{*} \neq 0 } | \cov( D/\gamma_k^{*}, Z_k)|\geq c$. \customlabel{assumption5}{A5}
\item[]A6. Assume $\eta_{max} = \max(1, B \sup_{\mathcal{S}:|\mathcal{S}|=s} \max_{1\leq j \leq s}  \|[  \Sigma_{\mathcal{S}}^{-1}]_{-j,j} \|_1) = o\big( (n/\log p)^{1/8}\big)$.
\customlabel{assumption6}{A6}

\item[] A7. Let 
$\lambda_{\gamma}= C_{\gamma}\surd{\left\{\left(\log p\right)/n\right\}}$ with $C_{\gamma} \geq 5.7\surd{\left\{B \left(  \alpha_D^{*\T}\Sigma_U   \alpha_D^* + \sigma_D^2\right)\right\}}$ and 
$ \lambda_{\Gamma}=C_{\Gamma}\surd{\left\{\left(\log p\right)/n\right\}}$ with $C_{\Gamma} \geq 5.7\surd{\left[B\left\{\left(\beta^*  \alpha_D^{*}+  \alpha_Y^*\right)^{\T}\Sigma_U \left(\beta^*  \alpha_D^{*}+  \alpha_Y^*\right) + \beta^{*2}\sigma_D^2 + \sigma_Y^2\right\}\right]}$, where $\lambda_{\gamma}, \lambda_{\Gamma}$ are tuning parameters in the lasso problems defined after \eqref{debias_sqrt_lasso} in Section \ref{M_construct} of the supplement.  Moreover, 
the tuning parameters in the nodewise lasso problems \eqref{nodewise_lasso} satisfy $\lambda_j \equiv  4a_2 \eta_{max}\surd{\left\{\left(\log p\right)/n\right\}} $ for $j = 1, \ldots, s$ and a constant $a_2>4\surd{3} eB\kappa^2$. \customlabel{assumption7}{A7}
\end{itemize}

Discussion of these assumptions are provided in Section \ref{dissAssump} of the supplement. 
We first consider the concentration of Wald ratios from spurious instruments. 

\begin{theorem}[Concentration for spurious instruments]\label{IrrConc_highD}
Suppose assumptions \ref{assumption1}--\ref{assumption7} hold. 
Further assume  $\log p=O\left(n^{\tau_1}\right)$ for $0<\tau_1<1$, $ n^{\tau_2} =O\left(s\right)$ for $0<\tau_2 <1$ and $\left(\log p\right)^{3/4} = o\left(n^{1/4} \surd{\log s}\right)$. 
Let $C_* = \cov\left(  \alpha_D^{*\T}   U,   \alpha_Y^{*\T}   U\right)/\var\left(  \alpha_D^{*\T}   U + \epsilon_D\right)$, and define 
$
\tilde C =  8\left\{\var\left(  \alpha_Y^{*\T}   U\right) -\cov\left(  \alpha_Y^{*\T}   U,   \alpha_D^{*\T}   U\right) C_* + \sigma_Y^2  \right\}^{1/2} \var^{-1/2}\left(  \alpha_D^{*\T}   U + \epsilon_D\right) $. 
Then there exists a constant $C_1 > 0$ independent of $\omega$ such that the following holds for $n,p$ sufficiently large,  
\[
\pr\left( \cup_{l \in  \check{\mathcal{I}} \cap \widehat{\mathcal{S}}_2} \left\{ \left| \frac{\widehat{\Gamma}_{\widehat{\mathcal{S}}_1,l}}{\widehat{\gamma}_{\widehat{\mathcal{S}}_1,l}} - \beta^* - C_* \right| > \tilde C\omega^{-1/2}   \right\}  \right) \leq 2s^{-3}
+ e^{- C_1\log(n \wedge p)}.
\]
\end{theorem}
Theorem~\ref{IrrConc_highD} formally verifies the empirical observation that causal effect estimates from spurious instruments concentrate in a region around $\beta^*+C_*$. The width of this region shrinks for higher threshold parameters $\omega$.

The center $\beta^* + C_*$ coincides with the least squares estimate, and also appears in the weak instrument literature \citep[e.g.][]{Nelson1990}, 
which suggests that as the instrumental strength approaches zero, the causal effect estimate approaches $\beta^* + C_*$. In contrast to these results, we allow the number of irrelevant variables
before screening to grow at an exponential rate and explicitly take into account randomness from both stages: first-stage screening and hard thresholding based on the de-biased lasso. Our theoretical analysis is hence much more delicate.

The next theorem formalizes the behavior of Wald ratios from relevant  instruments that pass joint thresholding. This applies to both valid and invalid instruments.

\begin{theorem}[Concentration results for relevant instruments]\label{RelConc_highD}
Suppose \ref{assumption1}, \ref{assumption2} and \ref{assumption4}--\ref{assumption7}
hold. Assume 
$\log p=O\left(n^{\tau_1}\right)$ for $0<\tau_1<1$ and $ n^{\tau_2} =O\left(s\right)$ for $0<\tau_2 <1$. 
Then the set $\widehat{\mathcal{S}}_2$ includes all relevant variables with probability going to one. 
Moreover, there exist constants $C_1, C_2>0$ such that the following holds when $n,p$ are sufficiently large:
\begin{align*}
&\quad \pr\left( \max_{l \in  \check{\mathcal{R}} \cap \widehat{\mathcal{S}}_2 } \left| \frac{\widehat{\Gamma}_{\widehat{\mathcal{S}}_1,l}}{\widehat{\gamma}_{\widehat{\mathcal{S}}_1,l}} -\left( \beta^* + \frac{\pi_{\widehat{\mathcal{S}}_1, l}^*}{\gamma_{\widehat{\mathcal{S}}_1, l}^*}\right) \right|> C_2 \frac{ \eta_{max}}{C_L} \surd{\left(\frac{\log p}{n}\right)}  \right) \leq e^{-C_1 \log (n\wedge p) }.
\end{align*}
\end{theorem}

Theorem~\ref{RelConc_highD} confirms that estimates from valid instruments concentrate around the true causal effect $\beta^*$ and estimates from invalid instruments 
concentrate around $\beta^* + \pi_{\widehat{\mathcal{S}}_1, l}^*/\gamma_{\widehat{\mathcal{S}}_1, l}^* $, even after marginal screening and joint thresholding applied to the de-biased lasso estimates.

Combined with Theorem \ref{IrrConc_highD}, Theorem \ref{RelConc_highD} further suggests that estimates 
from valid instruments can be separated from spurious instruments with probability going to one for sufficiently large $\omega$ satisfying $\tilde C\omega^{-1/2} < |C_*|$.  However, some valid instruments may be removed if $\omega$ is too large. We further discuss this point later in Section \ref{estimation}.

Results in this section formalize the intuition for why a naive application of the procedure to identify valid instruments under the plurality rule after applying marginal screening can fail--all estimates from spurious instruments are concentrated, leading to a violation of the empirical plurality rule. Our results also suggest a way out: since valid and invalid instruments still concentrate around the same values as without marginal screening, and since all estimates from spurious instruments are similar and separated from valid instruments, all that we need is to identify estimates from spurious instruments. This motivates the methodology in the next section.

\section{Fighting Noise with Noise}\label{estimation}

\subsection{Estimation of Spurious Instruments}

As established in the previous sections, a major challenge for causal inference with many candidate instruments is the presence of spurious instruments and distinguishing them from the valid ones.  
To solve this problem, 
the key idea in our developments is to generate independent noises, known as pseudo variables. These pseudo variables mimic the behaviour of the irrelevant variables in the original set. Suppose we independently generate the same number of pseudo variables as the original candidate IVs, 
and then perform marginal screening, joint estimation and joint thresholding with the concatenation of pseudo and original candidates. Let  $\check{\mathcal{S}}_1$, $\check{\mathcal{S}}_2$, $\bar{\mathcal{R}}$ and $\bar{\mathcal{I}}$ denote the counterpart of index sets $\widehat{\mathcal{S}}_1$ (candidates after marginal screening), $\widehat{\mathcal{S}}_2$ (candidates passing joint thresholding), $\check{\mathcal{R}}$ (relevant IVs passing marginal screening) and $\check{\mathcal{I}}$ (noises passing marginal screening), respectively. 
The concentration results for spurious instruments and relevant instruments when pseudos are included are formalized as Corollary \ref{concentration_PseudosInclude}.

\begin{corollary}\label{concentration_PseudosInclude}
Suppose same assumptions as in Theorem \ref{IrrConc_highD} and Theorem \ref{RelConc_highD}. If independently generated 
pseudo variables $\bar{\Z} \in \mathbb{R}^{n\times p}$ satisfy distribution assumptions A2 and A6, then with large probability, we have 
\begin{align}
& \left| \frac{\widehat{\Gamma}_{\check{\mathcal{S}}_1,l}}{\widehat{\gamma}_{\check{\mathcal{S}}_1,l}} - \beta^* - C_* \right| \leq \tilde C\omega^{-1/2}  \text{ for all } l \in  \bar{\mathcal{I}} \cap \check{\mathcal{S}}_2, \label{concentration_spurious}\\
&  \left| \frac{\widehat{\Gamma}_{\check{\mathcal{S}}_1,l}}{\widehat{\gamma}_{\check{\mathcal{S}}_1,l}} -\left( \beta^* + \frac{\pi_{\check{\mathcal{S}}_1, l}^*}{\gamma_{\check{\mathcal{S}}_1, l}^*}\right) \right| = o\left( \left(\frac{\log p}{n} \right)^{3/8} \right) \text{ for all } l \in  \bar{\mathcal{R}} \cap \check{\mathcal{S}}_2. \label{concentration_relevant}
\end{align}
\end{corollary}
The proof of Corollary \ref{concentration_PseudosInclude} is provided in the supplementary Section \ref{app:proof}. Next we present the procedure for detecting spurious instruments in Algorithm \ref{alg:rmSpurious}.

\begin{algorithm}
\caption{Fighting Noise with Noise}\label{alg:rmSpurious}

\begin{tabbing}
\quad  Input: Design matrix $  \Z \in \mathbb{R}^{n \times p}$, observed exposure $  \D \in \mathbb{R}^n$,   observed outcome $  \Y \in \mathbb{R}^n$.\\
    \quad \enspace  0. Independently generate the same number of pseudo variables as $\Z$, termed as $\bar \Z \in \mathbb{R}^{n \times p}$. \\
    \qquad \enspace Concatenate $ \Z$ and $\bar \Z$ by column and denote by $\tilde{\Z} \in \mathbb{R}^{n \times 2p} $. \\
    \quad \enspace  1. Marginal screening: Select the top $s$ candidate instruments by ranking \\
 \qquad \enspace $
  |\sum_{i=1}^{n} D_i \tilde{Z}_{ij}|/\left(\sum_{i=1}^n \tilde{Z}_{ij}^2 \right)
$
for $1\leq j \leq p$. $\check{\mathcal{S}}_1 \subset \{ 1, \ldots, p \}$ denotes the index set of \\
    \qquad \enspace the top $s$ candidate IVs. \\

    \quad \enspace  2-1. Joint Estimates via debiased-lasso:    $
      \Y \sim \tilde{\Z}_{\check{\mathcal{S}}_1}\longrightarrow  \widehat{\Gamma}_{\check{\mathcal{S}}_1}$; $\D \sim \tilde{\Z}_{\check{\mathcal{S}}_1} \longrightarrow  \widehat{\gamma}_{\check{\mathcal{S}}_1}
      $ \\

    \quad \enspace  2-2. Joint Thresholding: $ \check{\mathcal{S}}_2=\left\{1\leq l \leq s: |\widehat{\gamma}_{\check{\mathcal{S}}_1, l}| \geq \delta_n\times \text{SE}(\widehat{\gamma}_{\check{\mathcal{S}}_1,l}) \right\}$,  \\
   \qquad \enspace $\widehat{ \mathcal{P}}_1$ and $ \widehat{\mathcal{P}}_2$ denote the sets of original noise variables and pseudos  in $\check{\mathcal{S}}_2$ , respectively.\\
    
    \quad \enspace  3. Removing spurious IVs: Calculate $\widehat \beta^{(l)} = \widehat{\Gamma}_{\check{\mathcal{S}}_1,l} /\widehat{\gamma}_{\check{\mathcal{S}}_1,l}$; \\
   \qquad \enspace  $\widehat{L} = \min_{j \in \widehat{\mathcal{P}}_2 } \widehat{\beta}^{(j)}$ and $\widehat{U}=\max_{j \in \widehat{\mathcal{P}}_2 } \widehat{\beta}^{(j)}$; \\
    
  \qquad \enspace  Remove $\tilde{Z}_{\check{\mathcal{S}}_1, l}$: $ l \notin \check{\mathcal{S}}_2$  or 
   $\widehat{\beta}^{(l)} \in \left[\widehat{L} - c\left(\widehat{U} - \widehat{L} \right), \widehat{U} + c\left(\widehat{U} - \widehat{L} \right) \right]$ for $l \in \check{\mathcal{S}}_2$. \\
  \qquad \enspace  Here $c>0$ is a calibration constant;   $\check{\mathcal{S}}_3$ denotes the remaining candidate set. 
    
\end{tabbing}
\end{algorithm}

As we demonstrated theoretically, larger values for $\omega$ in $\delta_n = \left\{\omega \log\left( \max(n,s)\right) \right\}^{1/2}$ can lead to easier separation between valid instruments and spurious ones. Meanwhile, with larger $\omega$, the joint thresholding filters out more variables. In finite samples, this may cause some valid instruments to be discarded, leading to efficiency loss in downstream analysis. For practical application, we use $\omega=2.01$, which is from the tail bound of a Gaussian distribution. 
Additionally, as for the calibration constant $c$ used for estimating the spurious region, it is recommended to start with several different values to find a range that leads to stable causal estimations and ensures an appropriate selection of $c$.

In practice, pseudo variables should have variances similar to the original instruments to effectively detect spurious instruments. In the simulations, we generate pseudo copies from centered normal distributions with variance matching that of the candidate instruments. 
In our data application, where instruments are SNPs encoded as 0, 1, or 2, this is achieved by randomly sampling from these values for each SNP,
with the sampling probabilities based on their frequencies as observed in the data.

\begin{remark}\label{remark:spuriousDim}
Spurious instruments may not pose a practical problem when the number of candidate instruments is moderate or the sample size is very large, as seen with datasets like the UK Biobank, which includes hundreds of thousands of samples. In such cases, the step of removing spurious instruments does not impact the final causal estimation.
However, the data size may still be limited when practitioners evaluate causal relationships within specific subgroups. Some examples from the literature are provided in the supplementary Section \ref{MRex_small_n}. 
In such scenarios, the impact of spurious instruments becomes more pronounced. 
Our method, which carefully mitigates these effects, ensures more reliable causal conclusions. 
We illustrate the effect of spurious instruments on causal inference validity across various $(p/n)$ ratios in Supplement Section \ref{sec:various_np}.
\end{remark}

\subsection{Estimation of Causal Effect}

Algorithm \ref{alg:rmSpurious} plays a key role in removing  irrelevant variables but can introduce selection bias into causal estimation and jeopardize the validity of inference. A classical approach to addressing these challenges is sample splitting. 
Divide data into two halves: $H_1$ of size $n_1$ and $H_2$ of size $n_2$. To accommodate sample splitting in the present context, we implement Algorithm \ref{alg:rmSpurious} on the first sample $H_1$.  
Then on the second sample $H_2$, we use a mode-finding algorithm to estimate valid IVs from the remaining candidates  obtained via $H_1$, followed by causal estimation and inference. This procedure is outlined in Algorithm \ref{alg:ssplit}. See details in Section \ref{sec:sample-split} of the supplement.

\begin{algorithm}

\caption{Causal Inference via Fighting Noise with Noise} \label{alg:ssplit}

\begin{tabbing}
\quad  Input: Design matrix $  \Z \in \mathbb{R}^{n \times p}$, observed exposure $  \D \in \mathbb{R}^n$,   observed outcome $  \Y \in \mathbb{R}^n$.\\
    \quad \enspace  00. Divide data into two halves: 
    $H_1$ with  $n_1$ observations $(\Z^{(1)}, \D^{(1)}, \Y^{(1)})$, and \\
    \qquad \enspace \enspace $H_2$ with $n_2$ observations $(\Z^{(2)}, \D^{(2)}, \Y^{(2)})$. \\

   \quad  0--3. Implement Algorithm \ref{alg:rmSpurious} on  $H_1$.  
   The resulting IV set is $\check{\mathcal{S}}_3$. \\
    
   \quad  The following steps proceed on $H_2$: \\
    \qquad   4. Joint estimation based on $\check{\mathcal{S}}_3$:  $ \D^{(2)} \sim   \Z^{(2)}_{\check{\mathcal{S}}_3}$ $\longrightarrow \widehat{ \gamma}_{\check{\mathcal{S}}_3}^{(2)}$;  $ \Y^{(2)} \sim   \Z^{(2)}_{\check{\mathcal{S}}_3}$ $\longrightarrow \widehat{ \Gamma}_{\check{\mathcal{S}}_3}^{(2)}$. \\
    
   \qquad   5. Mode-finding: \\
   \qquad \enspace \enspace For $l \in \check{\mathcal{S}}_3$, calculate ratio estimates 
    $\widehat{\beta}^{(2,l)} = \widehat{ \Gamma}_{\check{\mathcal{S}}_3,l}^{(2)}/\widehat{ \gamma}_{\check{\mathcal{S}}_3,l}^{(2)};
    $ \\
    
   \qquad \enspace \enspace For $j,l \in \check{\mathcal{S}}_3$, compute
    $
    \widehat{b}^{(2, j, l)} =  \widehat{\beta}^{(2,l)}- \widehat{\beta}^{(2,j)}
    $ and $\text{SE}\left( \widehat{b}^{(2, j, l)}\right)$;\\
    
   \qquad \enspace \enspace For $j \in  \check{\mathcal{S}}_3$, find the number of candidates leading to a similar ratio estimate with $j$: \\
   
 \qquad \qquad\qquad \qquad $
        V_j = \left\lvert \left\{l \in \check{\mathcal{S}}_3: \left| \widehat{b}^{(2, j, l)}\right| / \text{SE}\left( \widehat{b}^{(2, j, l)}\right) \leq  \left\{\left(\omega^2 \log n_2\right)\right\}^{1/2} \right\} \right\rvert;
    $ \\
    
   \qquad \enspace \enspace  Find  valid IVs using the plurality rule: 
    $
        \check{\mathcal{S}}_4 = \Big\{j:  V_j =\max_{l \in \check{\mathcal{S}}_3} V_l \Big\}.
    $\\
    
 \qquad   6. Causal effect estimation:  
    $
    \widehat{\beta} = \left[\left\{\widehat{ \gamma}_{\check{\mathcal{S}}_3, \check{\mathcal{S}}_4}^{(2)}\right\}^{\T}\widehat{ W} \widehat{ \Gamma}_{\check{\mathcal{S}}_3, \check{\mathcal{S}}_4}^{(2)} \right] \left[\left\{\widehat{ \gamma}_{\check{\mathcal{S}}_3, \check{\mathcal{S}}_4}^{(2)}\right\}^{\T}\widehat{ W} \widehat{ \gamma}_{\check{\mathcal{S}}_3, \check{\mathcal{S}}_4}^{(2)}\right]^{-1};
    $ \\
 
 \qquad \enspace \enspace variance estimate
    $
    \widehat{v}  =  \left\{\widehat{\Theta}_{11} + \widehat{\beta}^2 \widehat{\Theta}_{22} - 2\widehat{\beta}\widehat{\Theta}_{12} \right\} \left[\left\{\widehat{ \gamma}_{\check{\mathcal{S}}_3, \check{\mathcal{S}}_4}^{(2)}\right\}^{\T}\widehat{ W} \widehat{ \gamma}_{\check{\mathcal{S}}_3, \check{\mathcal{S}}_4}^{(2)}\right]^{-1};
    $ \\
     \qquad \enspace \enspace  $\widehat{ W}$ is given in  \eqref{WeightW}; $\widehat{\Theta}_{11}, \widehat{\Theta}_{22}, \widehat{\Theta}_{12} $  are given in \eqref{var_cov}. \\  
     
\qquad 7. $100(1-\alpha)\%$ confidence interval:  $
    \left[ \widehat{\beta} - z_{1-\alpha/2}n_2^{-1/2}\widehat{v}^{1/2}, \widehat{\beta} + z_{1-\alpha/2}n_2^{-1/2}\widehat{v}^{1/2} \right].
    $     
      
\end{tabbing}
\end{algorithm}

Sample splitting can address selection bias in theory, however, it also results in smaller samples for identifying valid instruments and for the final estimation. When data size is limited, 
the reduced sample size can lead to loss of efficiency and less precise inferences, potentially negating the advantages of mitigating selection bias. We observe this effect in some of our simulations, as reported in Section \ref{simulations} and Section \ref{add_mainSimu} of the supplement.  Furthermore, when the instruments are not continuous, such as in our data application where IVs are coded as 0, 1, 2, splitting the data can substantially alter the distribution of the instruments.
We recommend implementing the full sample procedure in practice, where Steps 00--5 of Algorithm \ref{alg:ssplit} are applied to the entire dataset.  Causal inference is then conducted by refitting two-stage least squares with the estimated valid instruments using the full data.
This approach is detailed in Section \ref{sec:proposedFull} of the supplement.
Empirical evidence from Sections \ref{simulations} and \ref{simulation_Supp} suggests that the causal estimate obtained by the refitted model is consistent and asymptotically normal.

\section{Theory}\label{sec:theory_procedure}

In this section, we provide the theoretical guarantee of Algorithm \ref{alg:rmSpurious} for removing irrelevant variables from the candidate instrument set, and asymptotic properties of the sample splitting-based causal effect estimate obtained from Algorithm \ref{alg:ssplit}. Without loss of generality, we restrict our focus on univariate unmeasured confounder.

\subsection{Detection of spurious instruments}\label{sec:thmSpuriousDetect}

Define 
$
T_l=\left\{\tilde{v}_l^{\T} (\alpha_D^* \U+\epsilon_D) \right\}^{-1}\left\{\tilde{v}_l^{\T} (\alpha_Y^* \U+\epsilon_Y)\right\} - C_* 
$
for $l=1, \ldots, s$, where  $\tilde{v}_l = \tilde{\Z}_{\check{\mathcal{S}}_1} \tilde{M}^{\T} e_l /n $; $\tilde{\Z}_{\check{\mathcal{S}}_1}$ is the design matrix of candidate IVs that pass the marginal screening; $e_l \in \mathbb{R}^s$ is a unit vector with the $l$th component equal to one and zero otherwise;  
$\tilde{M} \in \mathbb{R}^{s\times s}$ is an inverse estimate of $ \widehat{\tilde{\Sigma}}_{\check{\mathcal{S}}_1} = n^{-1}\left(\tilde{\Z}_{\check{\mathcal{S}}_1}^{\T} \tilde{\Z}_{\check{\mathcal{S}}_1} \right)$. 
See Section \ref{M_construct} of the supplement for explicit constructions of $\tilde{M}$. Let $Q_l = E(T_l|T_1, \ldots, T_{l-1}, \tilde{\Z}, \alpha_D^* \U+\epsilon_D)$ and $R_l = T_l - Q_l$. Denote the cardinality of $\widehat{\mathcal{P}}_1$ and $\widehat{\mathcal{P}}_2$ by $q_1$ and $q_2$, respectively. 

We need the following additional assumptions:
\begin{itemize}
\item[]A8. There exists  $\tilde{B}>0$ such that 
$$
P\left(\frac{\left\{\max_{1\leq l \leq s} E\left(T_l^2 \big| \tilde{\Z}, \alpha_D^* \U+\epsilon_D\right) \right\}^{1/2} }{\left\{\min_{1\leq l \leq s} E\left(R_l^2 \big| \tilde{\Z}, \alpha_D^* \U+\epsilon_D\right) \right\}^{1/2}} > \tilde{B} \right) \to 0. 
$$
\item[]A9. Assume $ q_2^{-1}=o_p(1)$ and $q_1 q_2^{-1}=O_p(1)$.  
\end{itemize}

\begin{theorem}\label{pseudoCoverONoise}
Assume A1--A9. If the calibration constant $c>0$ satisfies $ 1+ c>\tilde{B}$, then, with probability tending to one, we have 
\begin{equation}\label{prob_pseudoCoverONoise}
\left[ \min_{l_l \in \widehat{\mathcal{P}}_1 }\widehat{\beta}^{(l_1)} , \max_{l_l \in \widehat{\mathcal{P}}_1 }\widehat{\beta}^{(l_1)}\right] \subseteq 
\left[\widehat{L} - c\left(\widehat{U} - \widehat{L} \right), \widehat{U} + c\left(\widehat{U} - \widehat{L} \right) \right],
\end{equation}
where $  \widehat{L}  =  \min_{l_2 \in \widehat{\mathcal{P}}_2 }  \widehat{\beta}^{(l_2)}$ and $\widehat{U} = \max_{l_2 \in \widehat{\mathcal{P}}_2 } \widehat{\beta}^{(l_2)} $. 
\end{theorem}
Theorem \ref{pseudoCoverONoise} suggests that the ratio estimates from spurious original noise variables fall into the region constructed by spurious pseudo variables with high probability. The proof is collected in the supplementary Section \ref{app:proof}. We provide explanations for the reasonableness of assumptions A8 and A9 in Section \ref{sec:assumptionExplain}.

\subsection{Asymptotic properties of causal effect estimation}

Let $\mathcal{S}_*=\left\{1\leq j \leq p: \gamma_j^* \neq 0 \right\}$ and $\mathcal{V}_* = \left\{ j \in \mathcal{S}_*: \pi_j^*=0 \right\}$ be the sets of relevant instruments and valid instruments, respectively. To prove the causal estimate obtained from Algorithm \ref{alg:ssplit} is consistent and asymptotically normal, we assume the plurality rule as in \cite{Guo2018}. 

\begin{definition}[Plurality rule]
The plurality rule holds under models \eqref{modD_XZ} and \eqref{modY_XZ} if $|\mathcal{V}_*| > \max_{h\neq 0} \left| \left\{j \in \mathcal{S}_*: \pi_j^*/\gamma_j^* = h \right\}\right|$. 
\end{definition}

We make an additional assumption 
\begin{itemize}
\item[]A10. 
For $j \in \mathcal{S}_* \setminus \mathcal{V}_*$,  
$(\gamma_j^*)^{-1} \pi_j^* \notin (C_*-|C_*|, C_*+|C_*| ) $.

\end{itemize}
Assumption A10 ensures that the region created by pseudo variables does not include invalid instruments,  with high probability.  
We can relax A10 so that the constraint applies only to invalid instruments that are correlated with valid instruments, and our results will still hold.
Lemma \ref{consistency_validSet} shows the consistency of the estimated set of valid instruments $\check{\mathcal{S}}_4$ from Algorithm \ref{alg:ssplit}.

\begin{lemma}\label{consistency_validSet}
Suppose A1--A10 and assumptions on pseudo variables as in Corollary \ref{concentration_PseudosInclude}. Further assume $n_1 = \varrho n_2$ for $\varrho >0$. 
Under the plurality rule, 
if the calibration constant $c$ satisfies $ \tilde{B} -1  <  c < \left(|C_*| \omega^{1/2} \tilde{C}^{-1}  - 1\right)/2$, 
then with large probability, $ \check{\mathcal{S}}_4 = \mathcal{V}^*$. 
\end{lemma}
The lower and upper bounds placed on the calibration constant $c$ ensure that all spurious instruments can be removed via pseudo variables while preserving the valid instruments.  
Building on Lemma \ref{consistency_validSet}, Theorem \ref{infernc_beta_split} presents asymptotic properties of the causal estimate obtained  from Algorithm \ref{alg:ssplit}. 
Proof of Lemma \ref{consistency_validSet} and Theorem \ref{infernc_beta_split} are collected in Section \ref{app:proof} of the supplement. 

\begin{theorem}\label{infernc_beta_split}
Suppose the same assumptions as Lemma \ref{consistency_validSet}. Then we have
\begin{equation}\label{asymp_dist}
    n_2^{1/2} \left(\widehat{\beta} - \beta \right) \overset{d}{\to}  N\left(0,  \frac{\alpha_Y^{*2}\sigma_u^2 + \sigma_Y^2}{ \gamma_{\mathcal{V}_*}^{*\T} \left( \Sigma_{\mathcal{V}_*, \mathcal{V}_*}  - \Sigma_{\mathcal{V}_*, \mathcal{S}_*\setminus \mathcal{V}_*}\left\{ \Sigma_{\mathcal{S}_*\setminus \mathcal{V}_*, \mathcal{S}_*\setminus \mathcal{V}_*}\right\}^{-1} \Sigma_{\mathcal{S}_*\setminus \mathcal{V}_*, \mathcal{V}_*} \right) \gamma_{\mathcal{V}_*}^{*} }\right)
\end{equation}
as $n \to \infty$, where $\Sigma_{\mathcal{I}_1, \mathcal{I}_2}$ is formed by rows and columns of $\Sigma$ indexed by $\mathcal{I}_1$ and $\mathcal{I}_2$, respectively. 
Moreover, the confidence interval for $\beta$ 
has an asymptotic coverage of $ 100(1-\alpha)\%$, i.e.,
\begin{equation}\label{asymp_coverage}
P\left(  \widehat{\beta} - z_{1-\alpha/2}n_2^{-1/2}\widehat{v}^{1/2} \leq \beta \leq \widehat{\beta} + z_{1-\alpha/2}n_2^{-1/2}\widehat{v}^{1/2}   \right) \to 1-\alpha
\end{equation}
as $n \to \infty$, where $ z_{1-\alpha/2} $ is the $(1-\alpha/2)$-quantile of the standard normal distribution. Details of $\widehat{v}$ are provided in Section \ref{ssplit_betaHat} of the supplementary materials. 
\end{theorem}

\section{Simulation Studies}\label{simulations}

We evaluate the proposed method through simulations. The setting is the same as in the motivating example in Section \ref{challenges}, except that here we consider various values of the random error variance in the exposure $ \sigma_D^2 = 0,2,4$. 
We include the following methods for comparison: 
(1) ``Proposed (full)'': Proposed method using the full sample (Algorithm \ref{algorithm1});
(2) ``Proposed (split)'': Sample splitting method (Algorithm \ref{alg:ssplit});
(3) ``Proposed (cross-fitting)'': ``Proposed (split)" with sample splitting replaced by cross-fitting \citep{Chernozhukov2018};
(4) ``CIIV'': Confidence interval method by \cite{Windmeijer2019b}. 
(5) ``Knockoff'': Application of knockoff method \citep{Candes2018} for estimating relevant IVs based on model \eqref{modD_XZ}, followed by mode-finding on these IVs and two-stage least squares (2SLS) estimation; 
(6) ``MR-median'': Weighted median MR method \citep{Bowden2016}; 
(7) ``MR-mode'':  Mode-based MR method  \citep{Hartwig2017};
(8) ``MR-cml'': Constrained maximum likelihood-based MR method  \citep{Xue2021}; 
(9) ``Naive'': Naive algorithm in Section \ref{naive_combo};
(10) ``Oracle": 2SLS estimation based on valid instruments only. 
Implementation details are provided in Supplement Section \ref{implement_details}.

Motivated by the data application in Section \ref{real_data}, we take $s=500$ in the marginal screening and calibration constant $c=0.05$. We use $n_1=1.5 n_2$ in the sample splitting algorithm. 
In Table \ref{tblS1-1}, we present the bias, root-mean-square error (RMSE), and empirical coverage under $\sigma_D^2 = 0, 4$ over 1000 simulation replications. 
Results for $\sigma_D^2=2$ are reported in Supplement Table \ref{tblS1-1-sigmaD2}.

\begin{table}[htbp]
\centering
\caption{Performance summary for various methods and $\sigma_D^2=0,4$ across 1000 Monte Carlo runs with $n=500$ and $p=50,000$. 
Standard errors ($\text{SE}\times 10$) are presented in the bracket if applicable. The nominal coverage probability is 0.95 
}
\resizebox{\columnwidth}{!}{%
\begin{tabular}{>{\rowmac}c>{\rowmac}c>{\rowmac}c>{\rowmac}c>{\rowmac}c>{\rowmac}c>{\rowmac}c>{\rowmac}c<{\clearrow}}
\toprule
&&& $\sigma_D^2 = 0$&&&$\sigma_D^2 = 4$& \\
\midrule
Method &  & \begin{tabular}{@{}c@{}}$\text{Bias}\times 10$ \\ ($\text{SE}\times 10$)\end{tabular}   & $\text{RMSE}\times 10$  & \begin{tabular}{@{}c@{}}Coverage \\ (Nominal = 95\%) \end{tabular} &   \begin{tabular}{@{}c@{}}$\text{Bias}\times 10$ \\ ($\text{SE}\times 10$)\end{tabular}   & $\text{RMSE}\times 10$  & \begin{tabular}{@{}c@{}}Coverage \\ (Nominal = 95\%) \end{tabular} \\ 
\midrule
Proposed (full) & & $-$0.10(0.02) & 0.72 & 0.92 & $-$0.18(0.05) & 1.54 & 0.90 \\  
Proposed (split) & & $-$0.21(0.07) & 2.21 & 0.90 & 0.26(0.10) & 3.31 & 0.81 \\ 
Proposed (cross-fitting) 
& & $-$0.79(0.07) & 2.48& 0.83  &  $-$0.72(0.09) & 2.81 & 0.70 \\
CIIV & &  $-$7.34(0.02) & 7.37 & 0.00  & $-$5.92(0.01) & 5.94 & 0.00 \\   
Knockoff & &  $-$0.31(0.05) & 1.51 & 0.88  & $-$0.11(0.02) & 0.65 & 0.90 \\
MR-median  & &  1.67(0.02) & 1.79 & 0.09 &  0.78(0.02) & 1.03 & 0.72    \\
MR-mode  & & 1.26(0.09) & 3.09 & 0.03 &  0.90(0.04) & 1.65& 0.48     \\
MR-cml  & &  $-$8.61(0.05) & 8.77 & 0.03 & $-$6.64(0.05) & 6.81 & 0.07  \\
Naive & & $-$2.15(0.01) & 2.19 & 0.00  &  $-$1.33(0.02) & 1.44 & 0.07 \\  
Oracle && $-$0.02(0.01) & 0.28 & 0.94 & $-$0.02(0.01) & 0.27 & 0.94 \\
\bottomrule
\end{tabular}%
}
\label{tblS1-1}
\end{table}

Estimates from the ``Naive", ``CIIV" and MR methods are biased since many spurious instruments are misidentified as valid  in the causal estimations.   
When the variance $\sigma_D^2$ is small, ``Knockoff'' performs worse than ``Proposed(full)" and ``Proposed(split)", but much better than other competing methods. Both ``Proposed(full)" and ``Proposed(split)" exhibit small bias and acceptable empirical coverage. However, the RMSE of the sample splitting method is much larger.
For ``Proposed(split)", the RMSE increases and the coverage probability decreases further below the nominal level when $\sigma_D^2$ increases. The cross-fitting approach
fails to effectively compensate for the efficiency loss due to data splitting and performs less favorably compared to ``Proposed(split)" across all settings.

Moreover, when the variance $\sigma_D^2$ increases, the performance of ``Proposed(full)" remains relatively stable, indicating that our approach using the full data set still performs well even when the variance of unobserved confounding is relatively small compared to that of the random error $\epsilon_D$. As $\sigma_D^2$ increases, the performance of ``Knockoff" tends to improve and becomes similar to ``Proposed(full)".  
Additionally, while ``Knockoff" requires about 1.4 hours per run, our proposed method takes approximately 15 minutes. 
See detailed discussions in Section \ref{sec:maniSimulaton_add} of the supplement. 
Overall, ``Proposed(full)" consistently demonstrates performance close to oracle method  and is also computationally efficient.

To provide further insights, we plot the distribution of $ \widehat{\beta}^{(j)}$ 
collected from all Monte Carlo runs for valid, invalid, and spurious instruments in $ \check{\mathcal{S}}_2$, $\check{\mathcal{S}}_3$ and $\check{\mathcal{S}}_4$ respectively in (a)--(c) of Figure \ref{fig-S1X-sigma1}, obtained from ``Proposed (full)" under $\sigma_D^2=0$.
The concentration result with regard to valid instruments in Theorem \ref{RelConc_highD} suggests that $\widehat{\beta}^{(j)}$  for valid instruments are concentrated in an interval, and this interval is centered around $\beta^*$ and 
shrinks at certain rates as $n,p \to \infty$. However, this does not guarantee that  $\widehat{\beta}^{(j)}$s for valid instruments are centered around $\beta^*$, so our theorem is not contradicting 
Figure \ref{fig-S1X-sigma1} (a)--(c). 
We illustrate this using a simulation study with increasing $n$ in Section \ref{region_shrink_valid} of the supplement.
We can see that causal effect estimates from spurious instruments and valid instruments are still separable when $n$ is not large as it is in the simulation setup. Most spurious instruments lie in the region formed by causal effect estimates from pseudo variables and thus are removed by the proposed procedure. Thereafter
the collection of valid instruments satisfies the plurality rule, so the mode-finding step is able to identify the valid instruments. 
We consider additional settings in Section \ref{simulation_Supp} of the supplement.

\vspace{-0.1cm}
\begin{figure}[htbp]
\centering
\subfigure[$\check{\mathcal{S}}_2$]{\includegraphics[width=0.23\textwidth]{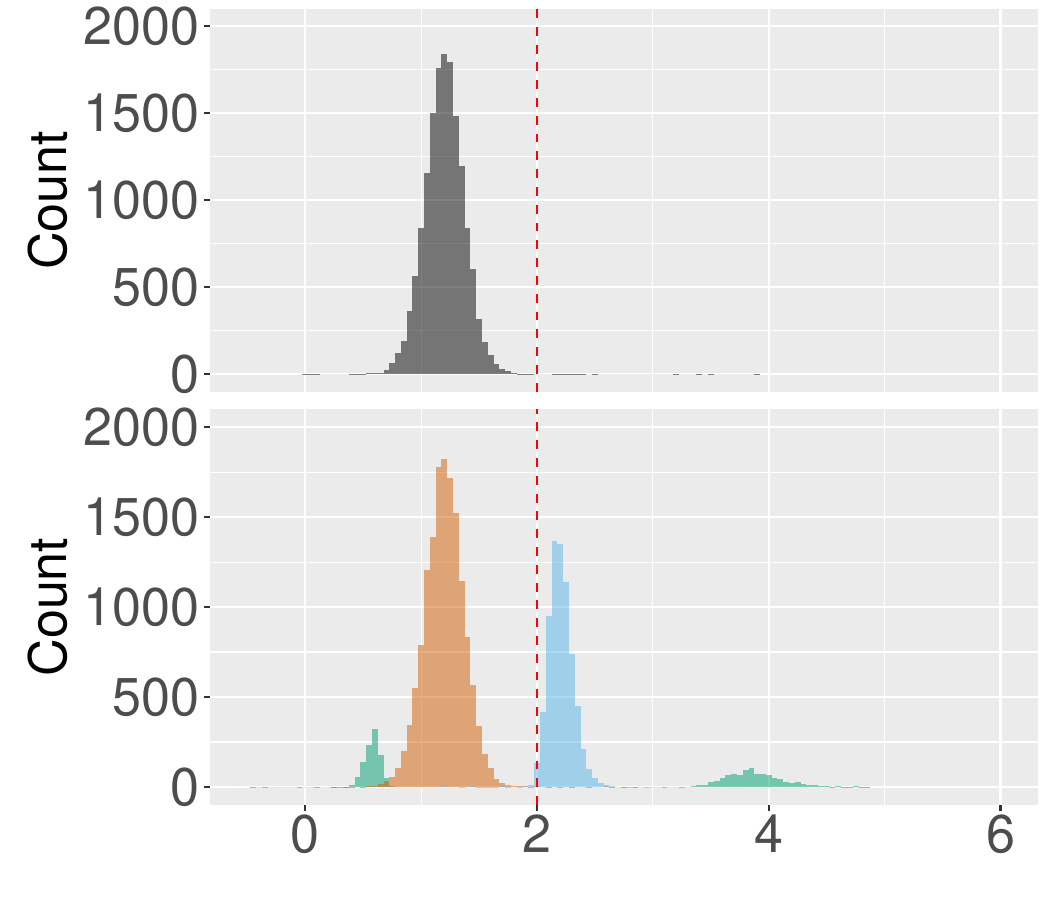}}
\hspace{0.2cm}
\subfigure[$\check{\mathcal{S}}_3$
]{\includegraphics[width=0.23\textwidth]{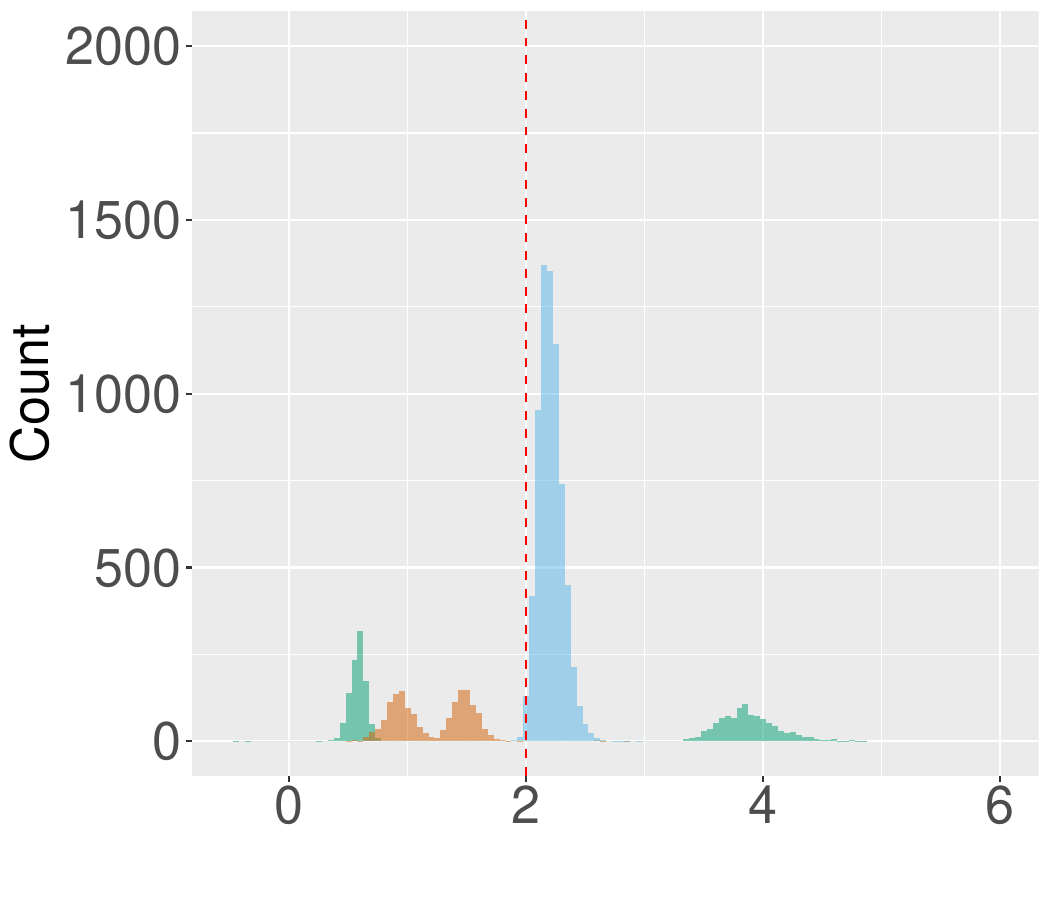}}
\hspace{0.2cm}
\subfigure[$\check{\mathcal{S}}_4$]{\includegraphics[width=0.23\textwidth]{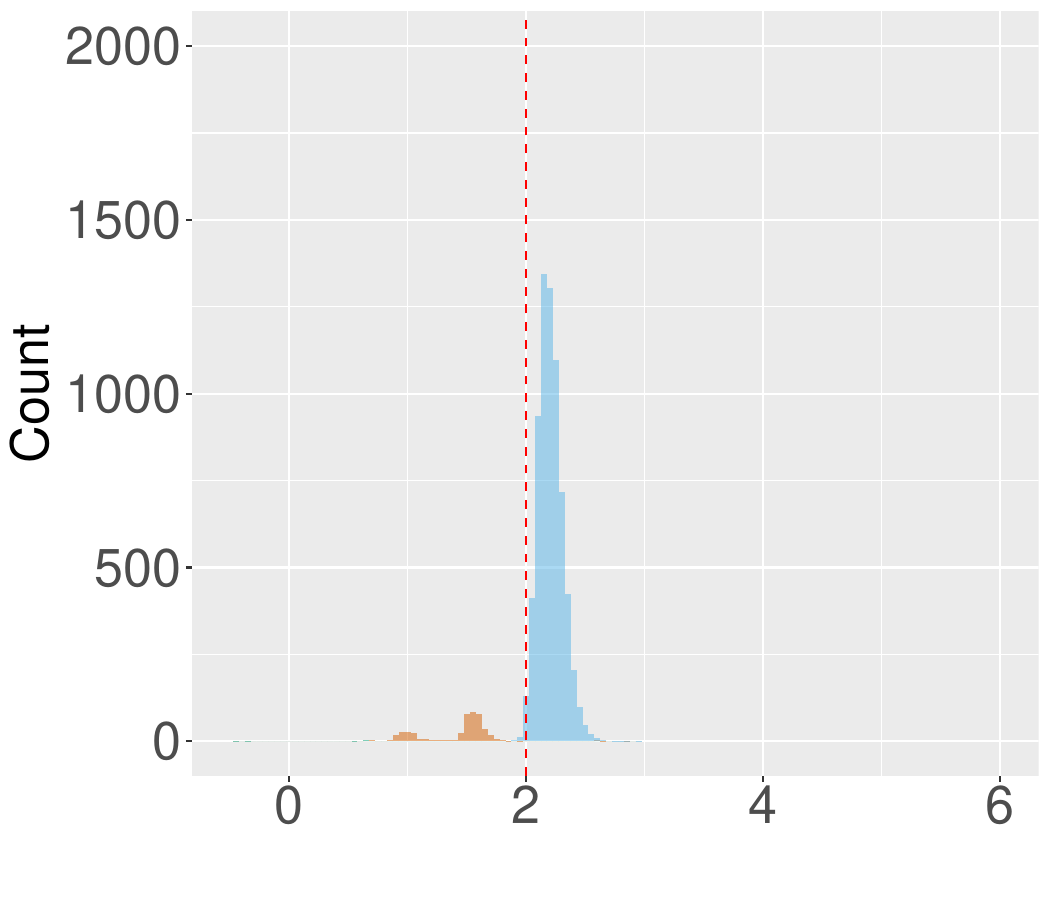}}
\hspace{0.2cm}
\subfigure[$\check{\mathcal{S}}_4$ (refit)]{\includegraphics[width=0.23\textwidth]{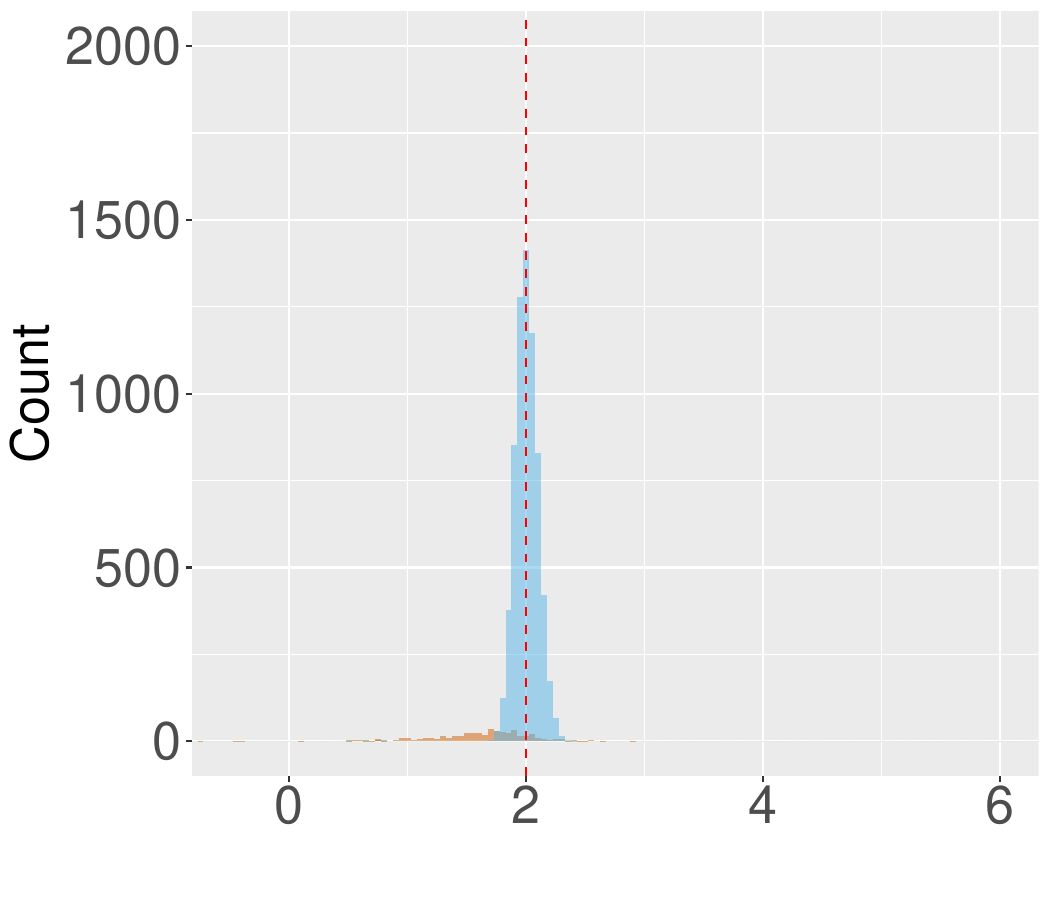}}
\vspace{-0.4cm}
\caption{Plots (a)-(d) correspond to the simulation in Section \ref{simulations}: 
Plots (a) to (c) are histograms of causal effect estimates from valid (blue), invalid (green), spurious (orange), and pseudo (black) instruments in $\check{\mathcal{S}}_2$, $\check{\mathcal{S}}_3$ and $\check{\mathcal{S}}_4$ respectively obtained from Algorithm \ref{algorithm1} across 1000 Monte Carlo runs. 
The red dashed line is the true causal effect $\beta^*=2$.
The results are based on $(n,p,\sigma_D^2)=(500, 50000, 0)$. Plot (d) gives the distribution of estimates from the refitted model.}
\label{fig-S1X-sigma1}
\end{figure}

\section{Real Data Application}\label{real_data}

We apply the proposed procedure to evaluate the effect of obesity on Health-Related Quality of Life (HRQL) using the data collected from the Wisconsin Longitudinal Study. 
See Section \ref{data_acknowledge} of the supplement for data acknowledgment.   
Our analysis focuses on unrelated graduates from Wisconsin high schools reinterviewed in 2011. We take BMI as exposure, a measure commonly used to define obesity. The healthy range for BMI is between 18.5 and 24.9. 
As shown in Figure \ref{BMI_HUI_scatter} in  the supplement, for underweight and normal individuals, 
HRQL is positively correlated with BMI, while for overweight or obese ones, HRQL is negatively correlated with BMI. 
We restrict our analysis to overweight participants whose BMI is above 25. There are 3023 subjects in total in our analysis. To measure HRQL, we use the Health Utility Index Mark 3 (HUI-3). We adjust for observed covariates including age, gender, years of education, and the top six principal components to account for population stratification \citep{NickPatterson2006}.
Detailed summary statistics and preliminary analysis are provided in Section \ref{data_prelim} of the supplement.

We begin with a crude analysis that ignores the unmeasured confounding. Simply regressing the HUI-3 on BMI gives an ordinary least squares (OLS) estimate  of $-$0.011 with 95\% CI [$-$0.013, $-$0.009], or equivalently, 20 units in BMI is associated with one standard deviation change in health-related quality of life; the standard deviation of  health-related quality of life is $0.227$. 
This effect seems too small to be plausible. 
Some unmeasured factors such as lifestyle may confound the relationship between BMI and HUI-3. In response, we use the Mendelian randomization method to investigate the causal relationship between these two. 
After quality control filtering, we get a candidate set of 3,683,868 SNPs for our  analysis. Details of the quality control are discussed in Supplement Section \ref{wls_preprocess}. 
We randomly sample from $\{0,1,2\}$ for each SNP according to the frequencies as they appear in the data, resulting in a total of 7,367,736 candidate variants. 
Then we apply marginal screening and keep the top $\floor{n/\log n} \approx 500$ \citep{JianqingFan2008} candidates, according to absolute marginal correlations between each candidate instrument and BMI. 
Those include  272 true variants and  228 pseudo ones. To further estimate relevant  SNPs to BMI, we fit a joint model for the reduced candidate set.

We visualize in Figure \ref{figWLS1}(a) the step-by-step estimation from the proposed method using the full dataset (Algorithm \ref{algorithm1}). A total of 81 SNPs ($\check{\mathcal{S}}_2(\text{Proposed})$) pass the joint thresholding, corresponding to all the points shown in Figure \ref{figWLS1}(a). The 43 points on the left of the orange dashed line correspond to true genetic variants, while the 38 red dots on the right are pseudo variants.  In $\check{\mathcal{S}}_2(\text{Proposed})$, causal effect estimates from 75 SNPs fall into the region formed by pseudo variants, which is represented by the band between the two black dashed lines. These SNPs are estimated as spurious and are removed from the candidate set. Here the spurious band is constructed using a calibration constant $c=0.05$. We conduct a sensitivity analysis over various values of $c$ in supplementary Section \ref{realData_more}. It is demonstrated that $c=0.05$ falls within a stable range for causal estimation, as shown in Figure \ref{fig:beta_vs_c}.  
The remaining 6 points outside the region are SNPs that are estimated to be relevant to BMI ($\check{\mathcal{S}}_3(\text{Proposed})$). An application of the mode-finding 
step on $\check{\mathcal{S}}_3(\text{Proposed})$ gives us 
4 valid instruments ($\check{\mathcal{S}}_4(\text{Proposed})$),  labelled as the blue squares in Figure \ref{figWLS1}(a). The genetic information of the estimated valid instruments is provided in Section \ref{dataApp_geneticInfo} of the supplement.
The purple cross mark at the top represents the estimated invalid instrument. 
Finally, we fit two-stage least squares using the three identified valid instruments and obtain the causal effect estimate $\widehat{\beta}=-0.040$ with 95\% confidence interval [$-$0.053, $-$0.028]. This suggests that in the overweight or obese population, one unit increase in BMI will result in 0.040 unit decrease of HUI-3. In other words, 5.8 units increase in BMI leads to roughly one standard deviations decrease in health-related quality of life.  
A sensitivity analysis of the uncertainty due to randomness in pseudo variables is included in the supplement Section \ref{realData_more}. 

\begin{figure}[htbp]
\centering
\subfigure[Proposed method using full data]{\includegraphics[width=4.6cm, height=3cm]{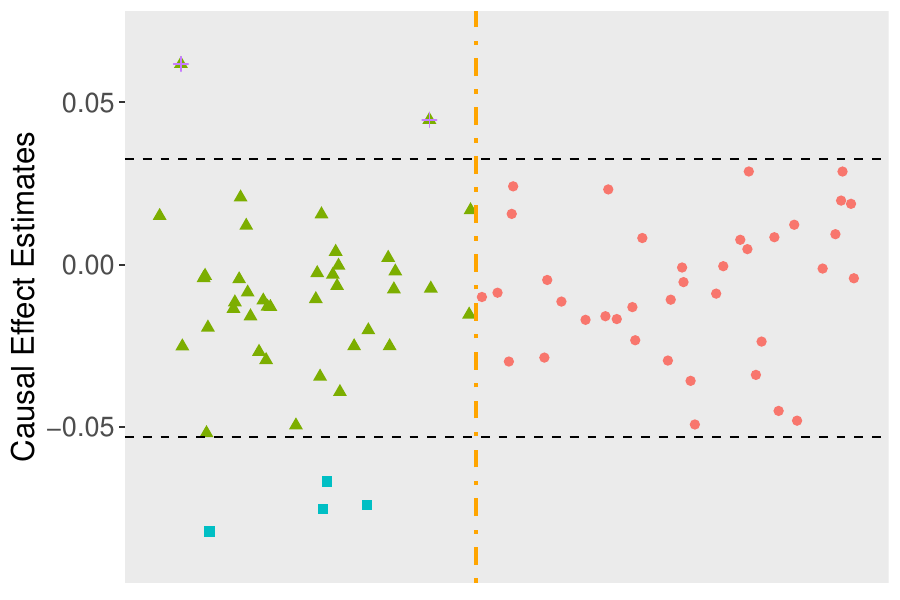}}
\quad
\subfigure[Naive procedure]{\includegraphics[width=4.2cm, height=3cm]{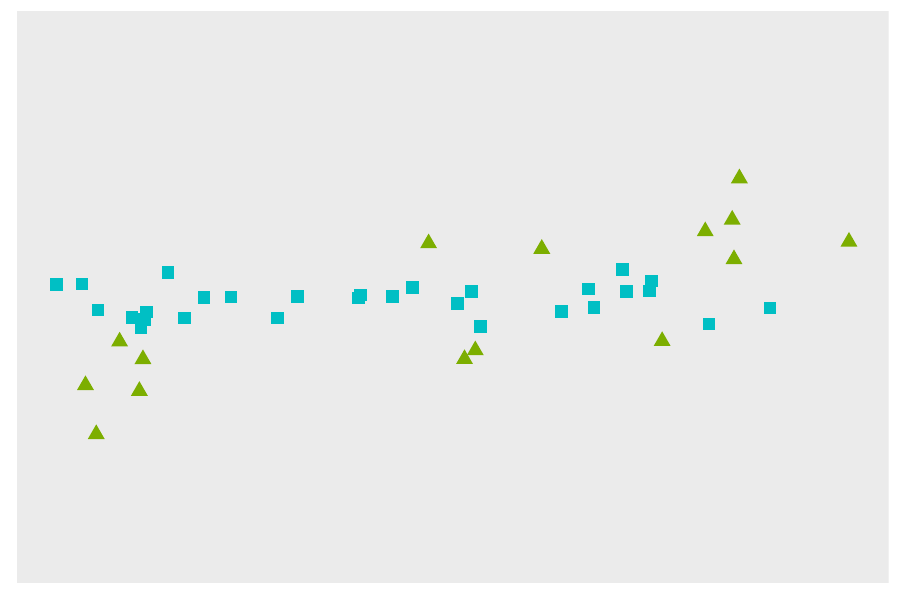}}
\quad
\subfigure[Naive procedure with pseudos]{\includegraphics[width=4.2cm, height=3cm]{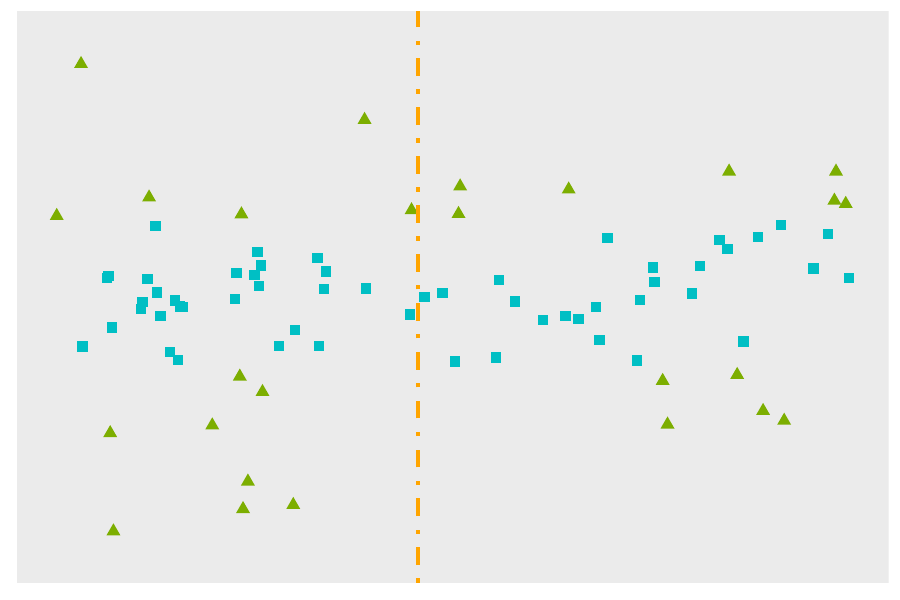}}
\vspace{-0.3cm}
\caption{Causal estimates from estimated relevant SNPs versus SNP index. (a) Causal estimates associated with SNPs in $\check{\mathcal{S}}_2(\text{Proposed})$. Orange dot-dashed line distinguishes pseudo from real SNPs. Red dots indicate pseudo SNPs; all others are true SNPs. Black dashed lines show the estimated range. Blue squares denote valid SNPs, and purple "+" signs represent invalid SNPs. (b) Causal estimates from relevant SNPs using the naive algorithm applied only to true SNPs. Blue squares and green triangles denote valid and invalid SNPs, respectively. (c) Results from the naive algorithm including pseudos. The orange dot-dashed line separates pseudo and real SNPs. Blue squares and green triangles denote valid and invalid SNPs, respectively.
}
\label{figWLS1}
\end{figure}

In our data application, we did not implement the sample splitting Algorithm \ref{alg:ssplit} for two reasons: first, the sample splitting method suffers from power loss, as demonstrated in our numerical experiments. Second, the SNP matrix is encoded as 0,1,2, which represents possible genotypes. Splitting the data might change the genotype distribution since some genotypes could appear more often in one portion of the data but not the other \citep{EmmanuelJ.Cands2019}.

To compare,  we apply the naive algorithm (``Naive") to this data set, but with true SNPs only. The 
estimated relevant instruments by ``Naive" are displayed in Figure \ref{figWLS1}(b), and the estimated valid instruments are labeled as blue squares. 
Interestingly, the resulting causal effect estimate  $-0.010$ (95\% CI [$-$0.015, $-$0.005]), is roughly the same as OLS estimate, equivalently saying that 20 units increase in BMI leads to approximately one standard deviation decrease in health-related quality of life.  The result is consistent with our theory that causal estimates from spurious instruments are similar  to the association estimate that ignores unmeasured confounding.
For further illustration, we also apply this method to the candidate set $\check{\mathcal{S}}_1 (\text{Proposed})$ obtained after marginal screening in our procedure.
In doing this, we pretend that the analyst does not know which variables in $\check{\mathcal{S}}_1 (\text{Proposed})$ are pseudo variables. 
Analysis results in Figure \ref{figWLS1}(c) show that as expected, ``Naive" misidentifies many pseudo variables as valid instruments.  
The resulting causal effect estimate,  $-0.008$ (95\% CI [$-$0.013, $-$0.004]), is close to OLS estimate.

For further comparison, we consider additional methods  including the confidence interval method (CIIV), median/mode/likelihood-based MR methods, which involve linkage disequilibrium (LD) clumping—a common preprocessing step to reduce dimension and decorrelate SNPs \citep[e.g.][]{Ye2021}. 
We also include comparisons with the knockoff-based method and the GWAS threshold. All these competing methods perform similarly to OLS. See detailed discussion in Section \ref{data:MR_CIIV} of the supplement.
In addition to this MR study, a potential application for evaluating online recommendation systems is discussed in Section \ref{app:tech} of the supplement.

\section{Discussion}\label{discussion}

Our work is related to a strand of research in econometrics regarding consistent estimation with many weak instruments \citep[e.g.][]{Chao2005}, and it can be extended to deal with weak instruments. The proposed method can be applicable under the two-sample design where the exposure and the outcome are collected from two independent samples. 
Additionally, our current framework only considers invalid instruments that have a direct effect on the outcome. Simulation evidence suggests that our method is robust to the violation of independence assumption (I3); see Section \ref{violate_indep} of the supplement for more details. 
Extending our work to nonlinear models, such as those with binary outcomes, represents another direction for future research.
Furthermore, using pre-screened instruments relevant to the exposure from a large external dataset might improve statistical power and is a promising area for continued exploration.
See Section \ref{detail_discuss} of the supplement for detailed discussions of these points.

\bibliographystyle{biometrika}
\bibliography{rangeTSHT}

\begin{thebibliography}{27}
\expandafter\ifx\csname natexlab\endcsname\relax\def\natexlab#1{#1}\fi

\bibitem[{Angrist \& Keueger(1991)}]{angrist1991does}
\textsc{Angrist, J.~D.} \& \textsc{Keueger, A.~B.} (1991).
\newblock Does compulsory school attendance affect schooling and earnings?
\newblock \textit{The Quarterly Journal of Economics} \textbf{106}, 979--1014.

\bibitem[{Barber \& Candes(2015)}]{RinaFoygelBarberan2015}
\textsc{Barber, R.~F.} \& \textsc{Candes, E.~J.} (2015).
\newblock Controlling the false discovery rate via knockoffs.
\newblock \textit{Ann. Statist.} \textbf{43}, 2055--2085.

\bibitem[{Barber \& Candes(2019)}]{EmmanuelJ.Cands2019}
\textsc{Barber, R.~F.} \& \textsc{Candes, E.~J.} (2019).
\newblock A knockoff filter for high-dimensional selective inference.
\newblock \textit{Ann. Statist.} \textbf{47}, 2504--2537.

\bibitem[{Bowden et~al.(2015)Bowden, Smith \& Burgess}]{Bowden2015}
\textsc{Bowden, J.}, \textsc{Smith, G.~D.} \& \textsc{Burgess, S.} (2015).
\newblock Mendelian randomization with invalid instruments: effect estimation
  and bias detection through {Egger} regression.
\newblock \textit{Int. J. Epidem.} \textbf{44}, 512--525.

\bibitem[{Bowden et~al.(2016)Bowden, Smith, Haycock \& Burgess}]{Bowden2016}
\textsc{Bowden, J.}, \textsc{Smith, G.~D.}, \textsc{Haycock, P.~C.} \&
  \textsc{Burgess, S.} (2016).
\newblock Consistent estimation in {Mendelian} randomization with some invalid
  instruments using a weighted median estimator.
\newblock \textit{Genet. Epidem.} \textbf{40}, 304--314.

\bibitem[{Burgess et~al.(2015)Burgess, Timpson, Ebrahim \& Smith}]{Burgess2015}
\textsc{Burgess, S.}, \textsc{Timpson, N.~J.}, \textsc{Ebrahim, S.} \&
  \textsc{Smith, G.~D.} (2015).
\newblock Mendelian randomization: where are we now and where are we going?
\newblock \textit{Int J Epidemiol.} \textbf{44}, 379--388.

\bibitem[{Candes et~al.(2018)Candes, Fan, Janson \& Lv}]{Candes2018}
\textsc{Candes, E.}, \textsc{Fan, Y.}, \textsc{Janson, L.} \& \textsc{Lv, J.}
  (2018).
\newblock {Panning for Gold: Model-X Knockoffs for High-dimensional Controlled
  Variable Selection}.
\newblock \textit{J. R. Statist. Soc. B.} \textbf{80}, 551--577.

\bibitem[{Carmi et~al.(2012)Carmi, Oestreicher-Singer \&
  Sundararajan}]{Carmi2012}
\textsc{Carmi, E.}, \textsc{Oestreicher-Singer, G.} \& \textsc{Sundararajan,
  A.} (2012).
\newblock {Is Oprah Contagious? Identifying Demand Spillovers in Online
  Networks}.
\newblock \textit{NET Institute Working Paper No. 10-18} .

\bibitem[{Chao \& Swanson(2005)}]{Chao2005}
\textsc{Chao, J.~C.} \& \textsc{Swanson, N.~R.} (2005).
\newblock Consistent estimation with a large number of weak instruments.
\newblock \textit{Econometrica} \textbf{73}, 1673--1692.

\bibitem[{Chernozhukov et~al.(2018)Chernozhukov, Chetverikov, Demirer
  et~al.}]{Chernozhukov2018}
\textsc{Chernozhukov, V.}, \textsc{Chetverikov, D.}, \textsc{Demirer, M.}
  et~al. (2018).
\newblock Double/debiased machine learning for treatment and structural
  parameters.
\newblock \textit{Econom. J.} \textbf{21}, C1--C68.

\bibitem[{Didelez \& Sheehan(2007)}]{VanessaDidelez2007}
\textsc{Didelez, V.} \& \textsc{Sheehan, N.} (2007).
\newblock Mendelian randomization as an instrumental variable approach to
  causal inference.
\newblock \textit{Stat Methods Med Res.} \textbf{16}, 309--330.

\bibitem[{Donoho \& Johnstone(1994)}]{Donoho1994}
\textsc{Donoho, D.~L.} \& \textsc{Johnstone, I.~M.} (1994).
\newblock Ideal spatial adaptation by wavelet shrinkage.
\newblock \textit{Biometrika.} \textbf{81}, 425--455.

\bibitem[{Fan \& Lv(2008)}]{JianqingFan2008}
\textsc{Fan, J.} \& \textsc{Lv, J.} (2008).
\newblock Sure independence screening for ultrahighdimensional feature space.
\newblock \textit{J. R. Statist. Soc. B.} \textbf{70}, 849--911.

\bibitem[{Guo et~al.(2018)Guo, Kang, Cai \& Small}]{Guo2018}
\textsc{Guo, Z.}, \textsc{Kang, H.}, \textsc{Cai, T.~T.} \& \textsc{Small,
  D.~S.} (2018).
\newblock Confidence intervals for causal effects with invalid instruments by
  using two‐stage hard thresholding with voting.
\newblock \textit{J. R. Statist. Soc. B.} \textbf{80}, 793--815.

\bibitem[{Hartwig et~al.(2017)Hartwig, Smith \& Bowden}]{Hartwig2017}
\textsc{Hartwig, F.~P.}, \textsc{Smith, G.~D.} \& \textsc{Bowden, J.} (2017).
\newblock Robust inference in summary data {Mendelian} randomization via the
  zero modal pleiotropy assumption.
\newblock \textit{Int J Epidemiol.} \textbf{46}, 1985---1998.

\bibitem[{Kang et~al.(2016)Kang, Zhang, Cai \& Small}]{Kang2016}
\textsc{Kang, H.}, \textsc{Zhang, A.}, \textsc{Cai, T.~T.} \& \textsc{Small,
  D.~S.} (2016).
\newblock Instrumental variables estimation with some invalid instruments and
  its application to {Mendelian} randomization.
\newblock \textit{J. Am. Statist. Ass.} \textbf{111}, 132--144.

\bibitem[{Kolesar et~al.(2015)Kolesar, Chetty, Friedman, Glaeser \&
  Imbens}]{Kolesar2015}
\textsc{Kolesar, M.}, \textsc{Chetty, R.}, \textsc{Friedman, J.},
  \textsc{Glaeser, E.} \& \textsc{Imbens, G.~W.} (2015).
\newblock Identification and inference with many invalid instruments.
\newblock \textit{J. Bus. Econ. Statist.} \textbf{33}, 474--484.

\bibitem[{Lawlor et~al.(2008)Lawlor, Harbord, Sterne, Timpson \&
  Smith}]{Lawlor2008}
\textsc{Lawlor, D.~A.}, \textsc{Harbord, R.~M.}, \textsc{Sterne, J. A.~C.},
  \textsc{Timpson, N.} \& \textsc{Smith, G.~D.} (2008).
\newblock Mendelian randomization: Using genes as instruments for making causal
  inferences in epidemiology.
\newblock \textit{Statist. Med.} \textbf{27}, 1133--1163.

\bibitem[{Nelson \& Startz(1990)}]{Nelson1990}
\textsc{Nelson, C.~R.} \& \textsc{Startz, R.} (1990).
\newblock Some further results on the exact small sample properties of the
  instrumental variable estimator.
\newblock \textit{Econometrica.} \textbf{58}, 967--976.

\bibitem[{Patterson et~al.(2006)Patterson, Price \& Reich}]{NickPatterson2006}
\textsc{Patterson, N.}, \textsc{Price, A.~L.} \& \textsc{Reich, D.} (2006).
\newblock Population structure and eigenanalysis.
\newblock \textit{PLoS Genet.} \textbf{2}, e190.

\bibitem[{van~de Geer et~al.(2014)van~de Geer, B{\"u}hlmann, Ritov \&
  Dezeure}]{VANDEGEER2014}
\textsc{van~de Geer, S.}, \textsc{B{\"u}hlmann, P.}, \textsc{Ritov, Y.} \&
  \textsc{Dezeure, R.} (2014).
\newblock On asymptotically optimal confidence regions and tests for
  high-dimensional models.
\newblock \textit{Ann. Statist.} \textbf{42}, 1166--1202.

\bibitem[{Wang \& {Tchetgen Tchetgen}(2018)}]{LinboWang2018}
\textsc{Wang, L.} \& \textsc{{Tchetgen Tchetgen}, E.} (2018).
\newblock Bounded, efficient and multiply robust estimation of average
  treatment effects using instrumental variables.
\newblock \textit{J. R. Statist. Soc. B.} \textbf{80}, 531--550.

\bibitem[{Windmeijer et~al.(2021)Windmeijer, Liang, Hartwig \&
  Bowden}]{Windmeijer2019b}
\textsc{Windmeijer, F.~A.}, \textsc{Liang, X.}, \textsc{Hartwig, F.~P.} \&
  \textsc{Bowden, J.} (2021).
\newblock The confidence interval method for selecting valid instrumental
  variables.
\newblock \textit{J. R. Statist. Soc. B.} \textbf{83}, 752--776.

\bibitem[{Wu et~al.(2007)Wu, Boos \& Stefanski}]{Wu2007}
\textsc{Wu, Y.}, \textsc{Boos, D.~D.} \& \textsc{Stefanski, L.~A.} (2007).
\newblock Controlling variable selection by the addition of pseudovariables.
\newblock \textit{J. Am. Statist. Ass.} \textbf{102}, 235--243.

\bibitem[{Xue et~al.(2021)Xue, Shen \& Pan}]{Xue2021}
\textsc{Xue, H.}, \textsc{Shen, X.} \& \textsc{Pan, W.} (2021).
\newblock Constrained maximum likelihood-based {Mendelian} randomization robust
  to both correlated and uncorrelated pleiotropic effects.
\newblock \textit{Am J Hum Genet} \textbf{108}, 1251--1269.

\bibitem[{Ye et~al.(2021)Ye, Shao \& Kang}]{Ye2021}
\textsc{Ye, T.}, \textsc{Shao, J.} \& \textsc{Kang, H.} (2021).
\newblock Debiased inverse-variance weighted estimator in two-sample
  summary-data {Mendelian} randomization.
\newblock \textit{Ann. Statist.} \textbf{49}, 2079--2100.

\bibitem[{Zhao et~al.(2020)Zhao, Wang, Hemani, Bowden \& Small}]{Zhao2020}
\textsc{Zhao, Q.}, \textsc{Wang, J.}, \textsc{Hemani, G.}, \textsc{Bowden, J.}
  \& \textsc{Small, D.~S.} (2020).
\newblock Statistical inference in two-sample summary-data {Mendelian}
  randomization using robust adjusted profile score.
\newblock \textit{Ann. Statist.} \textbf{48}, 1742--1769.

\end{thebibliography}

\end{document}